\documentclass[fleqn,10pt]{wlscirep}

\usepackage{slashed}
\usepackage{amsmath,amssymb,amsfonts,textcomp,amsthm}
\usepackage{graphicx}
\usepackage{amsthm}

\theoremstyle{definition}
\newtheorem{prop}{Proposition}

\title{Minimal control power of controlled dense coding and genuine tripartite entanglement}

\author[1]{Changhun Oh}
\author[1]{Hoyong Kim}
\author[1,2]{Kabgyun Jeong}
\author[1,*]{Hyunseok Jeong}
\affil[1]{Center for Macroscopic Quantum Control, Department of Physics and Astronomy, Seoul National University, Seoul, 08826, Korea}
\affil[2]{School of Computational Sciences, Korea Institute for Advanced Study, Seoul, 02455, Korea}
\affil[*]{h.jeong37@gmail.com}

\begin{abstract}
We investigate minimal control power (MCP) for controlled dense coding defined by the channel capacity.
We obtain MCPs for extended three-qubit Greenberger-Horne-Zeilinger (GHZ) states and generalized three-qubit $W$ states. Among those GHZ states, the standard GHZ state is found to maximize the MCP and so does the standard $W$ state among the $W$-type states. 
We find the lower and upper bounds of the MCP and show for pure states that the lower bound, zero, is achieved if and only if the three-qubit state is biseparable or fully separable. The upper bound is achieved only for the standard GHZ state. Since the MCP is nonzero only when three-qubit entanglement exists, this quantity may be a good candidate to measure the degree of genuine tripartite entanglement.
\end{abstract}

\begin{document}

\maketitle

\section*{Introduction}
Superdense coding \cite{bennett1992} is one of the simplest examples showing the power of quantum entanglement. It allows one to transmit two bits of classical information by sending only one qubit if a maximally entangled state is initially shared by the sender and receiver. It has been studied theoretically \cite{barenco1995, hausladen1996, bose2000, bowen2001, hiroshima2001, ziman2003, bruss2004, mozes2005} as well as experimentally \cite{mattle1996, fang2000, mizuno2005}.

Quantum teleportation \cite{bennett1993} is another intriguing example utilizing the power of quantum entanglement. It is a protocol to transmit an unknown quantum state using classical communication and an initially shared maximally entangled state. In fact, the equivalence of quantum teleportation and superdense coding with maximally entangled states has been proven \cite{werner2001}. 

Controlled dense coding \cite{hao2001, jing2003} and controlled quantum teleportation \cite{karlsson1998} have been proposed as extensions of superdense coding and quantum teleportation. The standard dense coding and quantum teleportation involve only two parties, a sender and a receiver, and assume that they share a maximally entangled state in advance. In controlled dense coding and teleportation, a third party participates in the protocol as a controller, and a tripartite entangled state is shared among the three parties. The controller in each protocol can control the channel capacity and the teleportation fidelity of the other two parties, respectively. Recently, the concepts of control power (CP) and minimal control power (MCP) of the controlled teleportation have been suggested to quantify how much teleportation fidelity can be controlled by the controller \cite{li2014, li2015, jeong2016, hop2016}. CP is defined as the difference between teleportation fidelities with and without the controller's assistance. MCP is defined as the minimal value of CP among all possible permutations of the three qubits. Similarly, CP has been defined in controlled remote state preparation schemes as well \cite{li2016}.

Motivated by the similarity between quantum teleportation and superdense coding, we here define CP and MCP of controlled dense coding of three-qubit states to quantify how much the dense coding channel capacity can be controlled by the controller. We also calculate CP and MCP for two representative tripartite entangled states: the extended Greenberger-Horne-Zeilinger (GHZ) state and the generalized $W$ state. We show that the standard GHZ state and the standard $W$ state have maximal MCP values for each class, which implies that MCP can be a candidate to capture the genuine tripartite entanglement of pure states. We also find the lower and upper bounds and analyze the properties in terms of genuine tripartite entanglement.

This paper is organized as follows: First, we review the controlled dense coding scheme and define CP and MCP for controlled dense coding. Then we calculate CP and MCP for important three-qubit states such as the extended GHZ states and generalized $W$ states. Furthermore, we investigate the properties of MCP in terms of genuine tripartite entanglement. Finally, we summarize and conclude the study.

\section*{Results}
\subsection*{Controlled dense coding and minimal control power}
\label{sec:cdc}
In the standard scenario of superdense coding \cite{bennett1992}, the sender, Alice, and the receiver, Bob, initially share a Bell state, and Alice encodes classical information by performing a local operation on her qubit and then sends it to Bob. After receiving her qubit, Bob performs a two-qubit measurement so that he can recover the classical information that Alice wants to transmit. Thus, Alice is able to transmit two classical bits by sending one qubit of a Bell state that is initially shared between them. Indeed, partially entangled states can be used in the dense coding scheme instead of maximally entangled states, and the number of bits that Alice can transmit to Bob using general two-qubit state $\rho_{12}$ can be quantified by the dense coding channel capacity \cite{bowen2001, bruss2004}
\begin{equation}
C(\rho_{12})=1+S(\rho_2)-S(\rho_{12})\label{cc1}
\end{equation}
where Alice has the first qubit, Bob the second, and $S(\rho)=-\text{Tr}(\rho \log_2\rho)$ is the von Neumann entropy. It is easy to confirm that the dense coding channel capacity of a Bell state is 2 bits. Note that, in this paper, $C(\rho_{jk})$ means the channel capacity when Alice has the $j$th qubit, Bob the $k$th qubit, and $C(\rho_{12})=C(\rho_{21})$ for a pure state $\rho_{12}$.

In the controlled dense coding scenario \cite{hao2001}, there is another party, Charlie, who plays the role as a controller, and the three parties share a three-qubit state $\rho_{123}$. To maximize the channel capacity between Alice and Bob, Charlie measures his qubit on an optimal basis that maximizes the average dense coding channel capacity and broadcasts the measurement outcome to Alice and Bob. After receiving the measurement outcome, Alice and Bob perform the dense coding scheme. We define the controlled dense coding channel capacity that Alice and Bob can achieve using this protocol as,
\begin{align}
C^{jkl}_{CD}&(\rho_{123})=\max_U[\langle 0|U\rho_j U^\dagger|0\rangle C(\rho^0_{kl})+\langle 1|U\rho_jU^\dagger|1\rangle C(\rho^1_{kl})], \label{cc2}
\end{align}
where Charlie, Alice, and Bob have the $j$th, $k$th, and $l$th qubit, respectively ($j,k,$ and $l$ are distinct numbers in $\{1,2,3\}$). The maximum is taken over all $2 \times 2$ unitary matrices $U$, which correspond to Charlie's choice of measurement basis. $\rho_j=\text{Tr}_{kl}(\rho_{123})$ and $\rho_{kl}^i$ is the quantum state that Alice and Bob share when Charlie's measurement outcome is $i$, where $i\in\{0,1\}$.

On the other hand, we now consider the channel capacity that Alice and Bob achieve without Charlie's assistance. In this case, the quantum state that Alice and Bob share is the reduced density operator of Alice and Bob, $\rho_{kl}=\text{Tr}_j (\rho_{123})$ where Alice and Bob have the $k$th and $l$th qubit, respectively. Thus, the channel capacity without Charlie's assistance is given by $C(\rho_{kl})$, which we will denote by $C_j(\rho_{kl})$ to specify Charlie's qubit. Finally, we define the CP of Charlie as the difference between the channel capacities with and without Charlie's assistance,
\begin{equation}
P^{jkl} (\rho_{123})\equiv C_{CD}^{jkl}-C_j(\rho_{kl})
\end{equation}
and also define the minimal control power (MCP) 
\begin{align}
P(\rho_{123})\equiv \min_{j,k,l}  P^{jkl} (\rho_{123})
\end{align}
where the minimum is taken over all possible permutation of $\{j,k,l\}$. MCP is defined in the same way as for controlled teleportation \cite{jeong2016} by replacing teleportation fidelity with channel capacity.

\subsection*{Examples}
\label{sec:ex}
Before we investigate the properties of MCP, let us calculate the MCPs of certain three-qubit states such as the extended GHZ states and generalized $W$ states \cite{acin2001}.
\subsubsection*{Extended GHZ states}
Let $|\psi_{\text{eGHZ}}\rangle_{123}$ be a state defined by
\begin{equation}
|\psi_{\text{eGHZ}}\rangle_{123}=\lambda_1|000\rangle+\lambda_2|110\rangle+\lambda_3|111\rangle,
\end{equation}
where the coefficients $\lambda_i\geq0$ and $\sum_i\lambda_i^2=1$ \cite{acin2001}. The state $|\psi_{\text{eGHZ}}\rangle$ is here called an extended GHZ state and $\psi_{\text{eGHZ}}=|\psi_{\text{eGHZ}}\rangle\langle\psi_{\text{eGHZ}}|_{123}$. Note that the extended GHZ states are invariant under the interchange of the first and second qubit.

To calculate CP, we need to obtain the channel capacities with and without Charlie's assistance for each permutation of $\{j, k, l\}$. The channel capacities with assistance can be obtained as follows:
\begin{align}
C^{123}_{CD}(\psi_{\text{eGHZ}})&=C^{132}_{CD}=C^{213}_{CD}=C^{231}_{CD}=1+h\bigg(\frac{1+\sqrt{1-4 \lambda_1^2\lambda_3^2}}{2}\bigg), \label{cc1} \\
C^{312}_{CD}(\psi_{\text{eGHZ}})&=C^{321}_{CD}=1+h(\lambda_1^2), \label{cc2}
\end{align}
where $h(x)\equiv-x \log_2 x-(1-x) \log_2 (1-x)$ is the binary entropy function. We have only two distinct values of the channel capacities with assistance because the extended GHZ states are invariant under the interchange of the first and second qubit and $C(\rho_{jk}^i)=C(\rho_{kj}^i)$ for pure states. An interesting fact about the extended GHZ states is that Charlie's measurement basis that maximizes the average channel capacity is always $\{(|0\rangle+|1\rangle)/\sqrt{2},(|0\rangle-|1\rangle)/\sqrt{2}\}$, regardless of the values of $\lambda_i$'s.

The channel capacities without assistance also can be obtained as:
\begin{align}
C_3(\rho_{12})&=C_3(\rho_{21})=1+h(\lambda_1^2)-h\bigg(\frac{1+\sqrt{1-4\lambda_1^2\lambda_3^2}}{2}\bigg), \label{c1} \\
C_1(\rho_{23})&=C_2(\rho_{13})=1-h(\lambda_1^2)+h\bigg(\frac{1+\sqrt{1-4\lambda_1^2\lambda_3^2}}{2}\bigg), \label{c2} \\
C_2(\rho_{31})&=C_1(\rho_{32})=1. \label{c3} 
\end{align}
Thus, the CPs are given by
\begin{align}
P^{123}(\psi_{\text{eGHZ}})&=P^{213}=h(\lambda_1^2), \\
P^{132}(\psi_{\text{eGHZ}})&=P^{231}=P^{312}=P^{321}=h\bigg(\frac{1+\sqrt{1-4 \lambda_1^2 \lambda_3^2}}{2}\bigg). 
\end{align}
Since $h\big(\frac{1+\sqrt{1-4 \lambda_1^2 \lambda_3^2}}{2}\big)\leq h(\lambda_1^2)=h\big(\frac{1+\sqrt{1-4 \lambda_1^2(\lambda_2^2+\lambda_3^2)}}{2}\big)$, MCP of the extended GHZ states is
\begin{equation}
P(\psi_{\text{eGHZ}})=h\bigg(\frac{1+\sqrt{1-4 \lambda_1^2 \lambda_3^2}}{2}\bigg) \label{ghz}
\end{equation}
which has the maximum value 1 if and only if $\lambda_1^2=1/2$ and $\lambda_3^2=1/2$, i.e., when the state is the standard GHZ state, $(|000\rangle+|111\rangle)/\sqrt{2}$. In addition, MCP is zero if and only if $\lambda_1=0$ or $\lambda_3=0$, i.e., the state is biseparable or fully separable \cite{acin2001-2,dur2000}. These facts imply that MCP is closely related to the genuine tripartite entanglement of pure states. To investigate the meaning of the result in terms of the tripartite entanglement, let us consider the three-tangle $\tau$ \cite{coffman2000, dur2000}, which is defined for a pure three-qubit state $|\psi\rangle_{123}$ as
\begin{equation}
\tau=\mathcal{C}^2_{j(kl)}-\mathcal{C}_{jk}^2-\mathcal{C}_{jl}^2,
\end{equation}
where $\mathcal{C}_{jk}=\mathcal{C}(\rho_{jk})=\mathcal{C}(\text{Tr}_l(\psi_{123}))$, $\mathcal{C}_{j(kl)}=\mathcal{C}(\psi_{j(kl)})$, $\psi_{123}=|\psi\rangle\langle\psi|_{123}$, and $\mathcal{C}$ is Wootters' concurrence \cite{hill1997,wootters1998}. The three-tangle for the extended GHZ states is calculated as
\begin{equation}
\tau=4 \lambda_1^2 \lambda_3^2.
\end{equation}
Thus, MCP is a monotonically increasing function of the three-tangle $\tau$,
\begin{equation}
P(\psi_{\text{eGHZ}})=h\bigg(\frac{1+\sqrt{1-\tau}}{2}\bigg).
\end{equation}
Indeed, the three-tangle $\tau$ is known as a value that quantifies genuine tripartite entanglement. Therefore, MCP can also be used to quantify the genuine tripartite entanglement for extended GHZ states. Note that MCP of the controlled teleportation is also a monotonically increasing function of the three-tangle \cite{jeong2016}.

In the special case that $\lambda_2=0$, the state $|\psi_{\text{eGHZ}}\rangle$ is reduced to generalized GHZ states,
\begin{equation}
|\psi_{\text{gGHZ}}\rangle_{123}=\lambda_1|000\rangle+\lambda_3|111\rangle.
\end{equation}
Using $\lambda_1^2+\lambda_3^2=1$, we can simplify channel capacities with Charlie's assistance in equations~\eqref{cc1},\eqref{cc2} as
\begin{align}
C^{ijk}_{CD}=1+h(\lambda_1^2),
\end{align}
and channel capacities without Charlie's assistance in equations~\eqref{c1},\eqref{c2},\eqref{c3} as
\begin{align}
C_i(\rho_{jk})=1,
\end{align}
for all $i,j,k$.
Thus, MCP of generalized GHZ states can be obtained as
\begin{equation}
P(\psi_{\text{gGHZ}})=h(\lambda_1^2).
\end{equation}
The difference between channel capacities with and without Charlie's assistance and MCP for generalized GHZ states are shown in Fig. \ref{capacity}. The figure shows how much channel capacity is affected by controller's assistance for generalized GHZ states.

For another special case ($\lambda_1=1/\sqrt{2}$),
\begin{equation}
|\psi_{\text{MS}}\rangle_{123}=\frac{1}{\sqrt{2}}|000\rangle+\lambda_2|110\rangle+\lambda_3|111\rangle,
\end{equation}
called the maximal slice states \cite{carteret2000}, MCP is given by
\begin{equation}
P(\psi_{\text{MS}})=h\bigg(\frac{1+\sqrt{2}\lambda_2}{2}\bigg).
\end{equation}

\begin{figure}
\centering
\begin{minipage}[h]{0.45\linewidth}
\includegraphics[width=\linewidth]{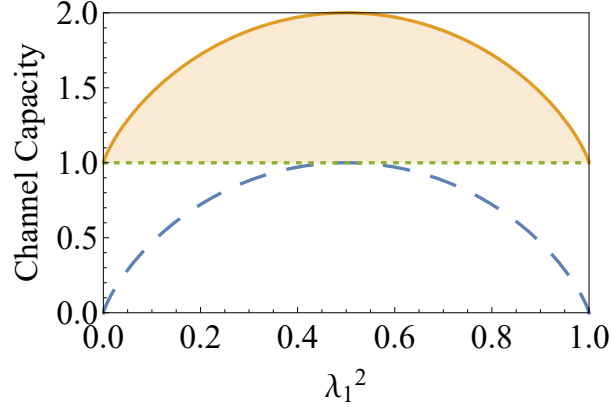}
\end{minipage}
\caption{\textbf{Channel capacities with and without Charlie's assistance and MCP for generalized GHZ states}. The solid curve represents channel capacities with Charlie's assistance, the dotted line those without Charlie's assistance, and dashed curve MCP of generalized GHZ states. The shaded region represents the amount of channel capacity controlled by Charlie.
}\label{capacity}
\end{figure}

\subsubsection*{Generalized $W$ states}
Let $|\psi_{\text{g}W}\rangle$ be a state defined by
\begin{equation}
|\psi_{\text{g}W}\rangle_{123}=\lambda_1|100\rangle+\lambda_2|010\rangle+\lambda_3|001\rangle,
\end{equation}
where the coefficients $\lambda_i\geq0$ and $\sum_i\lambda_i^2=1$ \cite{acin2001}. The state $|\psi_{\text{g}W}\rangle$ is here called a generalized $W$ state and $\psi_{\text{g}W}=|\psi_{\text{g}W}\rangle\langle\psi_{\text{g}W}|_{123}$.
The MCP of generalized $W$ states is calculated in the same way as before. First, the channel capacity with assistance is
\begin{equation}
C^{jkl}_{CD}(\psi_{\text{g}W})=1+(\lambda_k^2+\lambda_l^2)h\bigg(\frac{\lambda_1^2}{\lambda_k^2+\lambda_l^2}\bigg).
\end{equation}
In contrast to the case of extended GHZ states, the measurement basis that maximizes the channel capacity is $\{|0\rangle,|1\rangle\}$, regardless of the values of $\lambda_i$'s.
The channel capacity without assistance is
\begin{equation}
C_j(\rho_{kl})=1+h(\lambda_l^2)-h(\lambda_j^2).
\end{equation}
After simplification, the CPs can be obtained as
\begin{equation}
P^{jkl}(\psi_{\text{g}W})=(\lambda_j^2+\lambda_k^2)h\bigg(\frac{\lambda_j^2}{\lambda_j^2+\lambda_k^2}\bigg),
\end{equation}
and, thus, MCP can be written as
\begin{align}
P(\psi_{\text{g}W})&=\min_{j,k,l}\bigg[(\lambda_j^2+\lambda_k^2)h\bigg(\frac{\lambda_j^2}{\lambda_j^2+\lambda_k^2}\bigg)\bigg]  =(\lambda_j^2+\lambda_k^2)h\bigg(\frac{\lambda_j^2}{\lambda_j^2+\lambda_k^2}\bigg) \quad \text{if } \max\{\lambda_j^2,\lambda_k^2,\lambda_l^2\}=\lambda_l^2. 
\end{align}
It is easily checked that MCP of a generalized $W$ state is zero if and only if one $\lambda_i$ is zero, i.e., the state is biseparable or fully separable. In addition, we prove that the standard $W$ state, $(|100\rangle+|010\rangle+|001\rangle)/\sqrt{3}$, has maximal MCP among the generalized $W$ states as follows.

\begin{prop}
Let $\psi_{W}$ be the standard $W$ state, i.e., $\psi_{W}=\psi_{\text{g}W}$ with $\lambda_1=\lambda_2=\lambda_3=1/\sqrt{3}$. For any generalized $W$ state $\psi_{\text{g}W}$,
\begin{equation}
P(\psi_{\text{g}W})\leq\frac{2}{3}=P(\psi_{W}).
\end{equation}
\end{prop}

\begin{proof}
Suppose that there exists a generalized $W$ state such that $P(\psi_{\text{g}W})>\frac{2}{3}$, that is
\begin{equation}
(\lambda_j^2+\lambda_k^2)h\bigg(\frac{\lambda_j^2}{\lambda_j^2+\lambda_k^2}\bigg)>\frac{2}{3},
\end{equation}
for all distinct $j, k, l$ in $\{1,2,3\}$. Since the binary entropy is bounded from above by 1, i.e., $h(x)\leq 1$ for all $x\in[0,1]$, we have
\begin{equation}
\lambda_j^2+\lambda_k^2>\frac{2}{3},
\end{equation}
for all distinct $j, k, l$. This is, however, contradicting the normalization condition $\lambda_1^2+\lambda_2^2+\lambda_3^2=1$.
\end{proof}

\begin{figure}
\centering
\begin{minipage}[h]{0.9\linewidth}
\includegraphics[width=\linewidth]{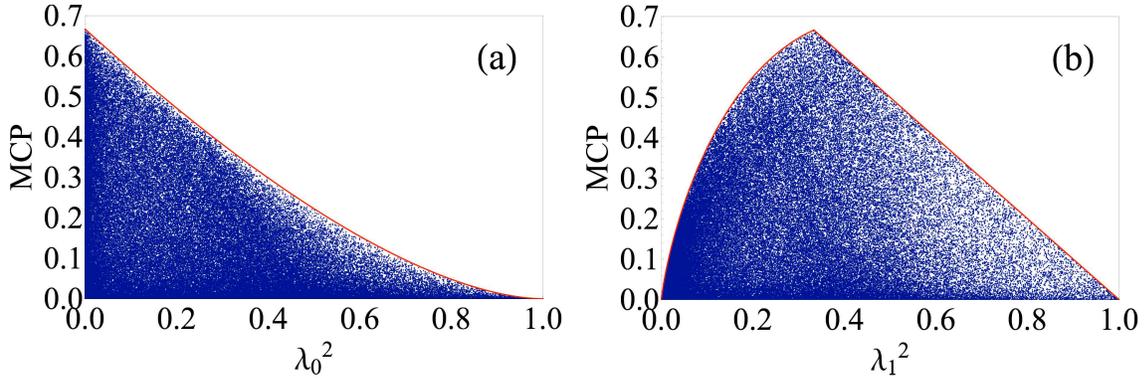}
\end{minipage}
\caption{\textbf{MCPs of randomly generated extended $W$ states.} (a) MCPs of randomly generated extended $W$ states plotted against $\lambda_0^2$. The solid curve corresponds to MCPs of extended $W$ states that have $\lambda_1=\lambda_2=\lambda_3$ with $\lambda_0^2$'s given. (b) MCPs of randomly generated extended $W$ states plotted against $\lambda_1^2$. The solid curve corresponds to the MCPs of extended $W$ states that have $\lambda_0=0$ and $\lambda_2=\lambda_3$ with $\lambda_1^2$'s given.
}\label{numerical}
\end{figure}

Let us now consider more general states, the extended $W$ states \cite{dur2000},
\begin{equation}
|\psi_{\text{e}W}\rangle_{123}=\lambda_0|000\rangle+\lambda_1|100\rangle+\lambda_2|010\rangle+\lambda_3|001\rangle,
\end{equation}
where $\lambda_i\geq0$ and $\sum_i\lambda_i^2=1$. Although MCP for a given extended $W$ state can be calculated easily, it is not easy to obtain an analytic expression for arbitrary extended $W$ states. Nevertheless, we generated $10^5$ extended $W$ states randomly and calculated their MCPs, and no extended $W$ state has been found that has a larger MCP than that of the standard $W$ state, $2/3$, which is shown in Fig.~\ref{numerical}. Furthermore, we can see from Fig.~\ref{numerical}(a) that for a given $\lambda_0^2$, no extended $W$ state has been found that has a larger MCP than that of the extended $W$ state with $\lambda_1=\lambda_2=\lambda_3$. We also conclude that MCP of the extended $W$ states with $\lambda_1=\lambda_2=\lambda_3$ is a monotonically decreasing function of $\lambda_0^2$. In addition, we also plot the same data against $\lambda_1^2$, which is shown in Fig.~\ref{numerical}(b). From the figure, we can also see that for a given $\lambda_1^2$, no extended $W$ state has been found that has a larger MCP than that of the extended $W$ state with $\lambda_2=\lambda_3$ and $\lambda_0=0$.

\subsection*{Properties of minimal control power} \label{sec:pro}
In this section, we investigate properties of MCP in terms of genuine tripartite entanglement. Before we start to introduce the properties of MCP, let us rewrite equation~\eqref{cc2} as

\begin{equation}
C^{jkl}_{CD}(\rho_{123})=1+\max_\mathcal{E} \bigg[\sum_{i=0}^1 p_i[S(\rho_l^i)-S(\rho_{kl}^i)]\bigg]
\end{equation}
where the maximum is taken over all possible ensembles $\mathcal{E}=\{p_i,\rho^i_{kl}\}$ satisfying $\sum_{i=0}^1 p_i \rho^i_{kl}=\rho_{kl}$, which corresponds to maximization over all Charlie's possible measurement basis \cite{hughston1993}. Thus, CP also can be written as
\begin{align}
P^{jkl}(\rho_{123})&=\max_\mathcal{E}\bigg[\sum_{i=0}^1 p_i[S(\rho_l^i)-S(\rho_{kl}^i)]\bigg]-S(\rho_l)+S(\rho_{kl})=\max_\mathcal{E}\bigg[\sum_{i=0}^1 p_i I(k\rangle l|\rho_{kl}^i)\bigg]-I(k\rangle l|\rho_{kl}), 
\end{align}
where we have used the definition of the coherent information, $I(k\rangle l|\rho_{kl})=S(\rho_l)-S(\rho_{kl})$, which is also known as a negative quantum conditional entropy \cite{wilde2013, nielsen2010}. Using the properties of the von Neumann entropy and convexity of coherent information, we can prove the following proposition.

\begin{prop} For general three-qubit states $\rho_{123}$, which can be mixed or pure,
\begin{equation}
0\leq P(\rho_{123}) \leq 1.
\end{equation}
\end{prop}

Before we start to prove the proposition, we review the inequalities for the von Neumann entropy of the mixture of quantum states $\rho_i$ \cite{nielsen2010},
\begin{equation}
\sum_i p_i S(\rho_i)\leq S\bigg(\sum_i p_i \rho_i \bigg)\leq \sum_i p_i S(\rho_i)+h(p_i), \label{entropy}
\end{equation}
where the equality of the first inequality holds if and only if all the states $\rho_i$ for which $p_i>0$ are identical, and the equality of the second inequality holds if and only if the states $\rho_i$ have support on orthogonal subspaces.

\begin{proof}
Let us prove the lower bound first. It can be shown that CP of $\rho_{123}$ is non-negative by using the convexity of the coherent information \cite{wilde2013, nielsen2010},
\begin{equation}
P^{jkl}(\rho_{123})=\max_\mathcal{E} \bigg[\sum_{i=0}^1 p_i I(k\rangle l|\rho_{kl}^i)-I(k\rangle l|\rho_{kl})\bigg]\geq 0. \label{convex}
\end{equation}
Thus, MCP is also non-negative, $P(\rho_{123})\geq 0$.
For the upper bound, using equation~\eqref{entropy}, we obtain for CP that
\begin{align}
P^{jkl}(\rho_{123})=\max_\mathcal{E} \bigg[\sum_{i=0}^1 p_i[S(\rho^i_l)-S(\rho^i_{kl})]\bigg]-S(\rho_l)+S(\rho_{kl})\leq\max_\mathcal{E}\bigg[-\sum_{i=0}^1 p_i S(\rho^i_{kl})\bigg]+S(\rho_{kl})\leq \max_\mathcal{E}h(p_i) \leq 1. \label{upper}
\end{align}
It is clear that MCP also has the same upper bound.
\end{proof}

We have found the lower and upper bounds of MCP of general three-qubit states. Now, let us investigate the conditions for pure states that attain the bounds by proving following two propositions.

\begin{prop}
For pure states $\rho_{123}$, $P(\rho_{123})=0$ if and only if $\rho_{123}$ is biseparable or fully separable.
\end{prop}

\begin{proof}
Let $\rho_{123}=|\psi\rangle\langle\psi|_{123}$. If $\rho_{123}$ is biseparable or fully separable, it is clear that when Charlie has the separated qubit, CP is zero, and thus MCP is zero. Let us now assume that $P(\rho_{123})=0$. Thus, there exists $\{j,k,l\}$ such that $P^{jkl}(\rho_{123})=0$. We set $\{j,k,l\}=\{3,1,2\}$ without loss of generality. Equation~\eqref{convex} implies that for any ensemble $\{p_i,\rho_{12}^i\}$ that satisfies $\sum_ip_i\rho_{12}^i=\rho_{12}$,
\begin{align}
\sum_i p_i I(1\rangle 2|\rho_{12}^i)-I(1\rangle 2|\rho_{12})=\sum_i p_iS(\rho_2^i)-S(\rho_2)+S(\rho_{12})=0
\end{align}
because the maximum is zero. Let us write the state in a Schmidt decomposition form $|\psi\rangle_{123}=\sum_i\sqrt{\mu_i}|\psi^i\rangle_{12}|i\rangle_3$, and let us assume that $\mu_i\neq 0$ because $|\psi\rangle_{123}$ is trivially biseparable or fully separable if $\mu_0=0$ or $\mu_1=0$. Then, it is easy to obtain
 that $S(\rho_{12})=h(\mu_0)$. Furthermore, an ensemble set $\{p_i=\mu_i,\rho_{12}^{i*}=|\psi^i\rangle\langle\psi^i|_{12}\}$ should satisfy that
\begin{align}
\sum_i\mu_iS(\rho_2^{i*})-S(\rho_2)+h(\mu_0)=0
\end{align}
where $\rho_2^{i*}=\text{Tr}_1(\rho_{12}^{i*})$. Therefore, $\rho_2^{0*}$ and $\rho_2^{1*}$ are orthogonal; thus we denote $\rho_2^{0*}=|0\rangle\langle0|_2$ and $\rho_2^{1*}=|1\rangle\langle1|_2$ without loss of generality. Thus, the state can be rewritten as $|\psi\rangle_{123}=\sqrt{\mu_0}|\phi^0\rangle_1|00\rangle_{23}+\sqrt{\mu_1}|\phi^{1}\rangle_1|11\rangle_{23}$ where $\langle0|1\rangle=0$, but $|\phi^{0}\rangle$ and $|\phi^{1}\rangle$ are not necessarily orthogonal. Then, $S(\rho_2)=h(\mu_0)$ so that $\sum_ip_iS(\rho_2^i)=0$ for any ensemble that satisfies $\sum_ip_i\rho_{12}^i=\rho_{12}$. Now, let us choose another ensemble set $\{p_i^{**}=1/2,\rho_{12}^{i**}=|\phi^{\pm}\rangle\langle\phi^{\pm}|_{12}\}$ where $|\phi^{\pm}\rangle_{12}=\sqrt{\mu_0}|\phi^0\rangle_1|0\rangle_2\pm\sqrt{\mu_1}|\phi^1\rangle_1|1\rangle_2$. Since $\sum_i p_i^{**}S(\rho_2^{i**})=0$ should be satisfied, we require that $|\phi^0\rangle_1=|\phi^1\rangle_1$. Thus, the first qubit is separated, i.e., $|\psi\rangle_{123}$ is biseparable.
\end{proof}

\begin{prop}
For pure states $\rho_{123}$, $P(\rho_{123})=1$ if and only if $\rho_{123}$ is the standard GHZ state.
\end{prop}

\begin{proof}
We have already shown that MCP of the standard GHZ state is 1 from equation~\eqref{ghz}. Now, let us assume that MCP of a pure three-qubit state $\rho_{123}$ is 1 so that the CPs of $\rho_{123}$ for any $\{j,k,l\}$ are 1. Note that since $S(\rho_{kl}^i)=0$ for pure states $\rho_{123}$, CP can be written as $P^{jkl}=\max_\mathcal{E}[\sum_i p_iS(\rho_l^{i})]-S(\rho_l)+S(\rho_{kl})$. To prove that a pure state $\rho_{123}$ that attains the upper bound should be the standard GHZ state, we need to prove that the measurement basis of Charlie that maximizes $P^{jkl}$ also maximizes $P^{jlk}$ and vice versa. Let $\mathcal{E}^*=\{p_i^*,\rho^{i*}_{kl}\}$ be an ensemble that maximizes $\sum_i p_iS(\rho_l^{i})$ and $\rho_l^{i*}=\text{Tr}_k (\rho_{kl}^{i*})$, and let $\mathcal{E}^{**}=\{p_i^{**},\rho^{i**}_{kl}\}$ be an ensemble that maximizes $\sum_i p_iS(\rho_k^{i})$ and $\rho_k^{i**}=\text{Tr}_l (\rho_{kl}^{i**})$. Using that $S(\rho_l^{i*})=S(\rho_k^{i*})$ and $S(\rho_l^{i**})=S(\rho_k^{i**})$ for pure states $\rho_{kl}^{i*}$ and $\rho_{kl}^{i**}$, we have
\begin{align}
&\sum_{i=0}^1p_i^{**}S(\rho_l^{i**})\leq\sum_{i=0}^1p_i^*S(\rho_l^{i*})=\sum_{i=0}^1p_i^*S(\rho_k^{i*}), \\
&\sum_{i=0}^1p_i^*S(\rho_k^{i*})\leq\sum_{i=0}^1p_i^{**}S(\rho_k^{i**})=\sum_{i=0}^1p_i^{**}S(\rho_l^{i**}).
\end{align}
Thus, we have $\sum_ip_{i}^*S(\rho_l^{i*})=\sum_ip_{i}^{**}S(\rho_l^{i**})$ and $\sum_ip_{i}^*S(\rho_k^{i*})=\sum_ip_i^{**}S(\rho_k^{i**})$, which means that both ensembles $\mathcal{E}^*$ and $\mathcal{E}^{**}$ maximize $\sum_i p_iS(\rho_k^{i})$ and $\sum_i p_iS(\rho_l^{i})$. Therefore, the measurement basis of Charlie that maximizes $P^{jkl}$ also maximizes $P^{jlk}$ and vice versa.

Let us start to prove that if $P(\rho_{123})=1$, where $\rho_{123}=|\psi\rangle\langle\psi|_{123}$, then $\rho_{123}$ is the standard GHZ state. Since the equality of the last inequality in equation~\eqref{upper} holds if and only if $p_i=1/2$, we can write $|\psi\rangle_{123}$, up to a local unitary transformation, as
\begin{equation}
|\psi\rangle_{123}=\frac{1}{\sqrt{2}}(|\psi^0\rangle_{12}|\phi_0\rangle_3+|\psi^1\rangle_{12}|\phi_1\rangle_3), 
\end{equation}
where we have chosen $\{j,k,l\}=\{3,1,2\}$ without loss of generality, and $\{|\phi_0\rangle, |\phi_1\rangle\}$ is a measurement basis of Charlie that maximizes both $P^{312}$ and $P^{321}$. In addition, $\langle \psi^0|\psi^1\rangle=0$ because the equality of the second inequality in equation~\eqref{upper} holds. Let us write $|\psi^0\rangle_{12}=\sum_i \sqrt{\lambda_i}|i\rangle_1|i\rangle_2$ and $|\psi^1\rangle_{12}=\sum_i \sqrt{\mu_i}|\bar{i}\rangle_1|\bar{i}\rangle_2$. Note that since the equality of the first inequality in equation~\eqref{upper} holds if and only if $\rho_l^0=\rho_l^1$, we have $\rho_2^0=\rho_2^1=\rho_2$. Since we can choose $\{j,k,l\}=\{3,2,1\}$, we also have $\rho_1^0=\rho_1^1=\rho_1$. Thus,
\begin{align}
\rho_1^0=&\sum_i \lambda_i |i\rangle_1\langle i|=\sum_i \mu_i|\bar{i}\rangle_1\langle\bar{i}|=\rho_1^1, \\
\rho_2^0=&\sum_i \lambda_i |i\rangle_2\langle i|=\sum_i \mu_i|\bar{i}\rangle_2\langle\bar{i}|=\rho_2^1.
\end{align}
Since the eigen-decomposition of $\rho_1$ and $\rho_2$ is unique, we require (i) $\lambda_i=\mu_i, |i\rangle_1=|\bar{i}\rangle_1$, and $|i\rangle_2=|\bar{i}\rangle_2$, or (ii) $\lambda_i=\mu_i=1/2$. However, the condition (i) contradicts $\langle\psi^0|\psi^1\rangle=0$. Thus, we can only require condition (ii), which gives us that $|\psi^0\rangle$ and $|\psi^1\rangle$ are two orthogonal Bell states. Hence, $|\psi\rangle_{123}$ is the standard GHZ state up to a local unitary operation.
\end{proof}

\section*{Discussion}
\label{sec:con}
We have considered controlled dense coding and defined its CP and MCP. We have also calculated MCPs for several important pure three-qubit states such as the extended GHZ states and generalized $W$ states. We found that MCP of the  extended GHZ states is a monotonically increasing function of the three-tangle of the states, which quantifies genuine tripartite entanglement. We also found that the standard GHZ state uniquely achieves the maximal value of MCP among the extended GHZ states, so does the standard $W$ state among the generalized $W$ states. 

Although the extended $W$ states have a genuine tripartite entanglement, the three-tangle of the extended $W$ states vanishes \cite{dur2000}. Based on the operational meaning of MCP, the three-party entanglement of the extended $W$ states can be witnessed by MCP as we have discussed in this paper. Thus, MCP can be used to quantify the three-party entanglement of the extended $W$ states.

We also found the lower and upper bounds of MCP for general three-qubit states. In addition, we proved that for pure states $\rho_{123}$, the equality of the lower bound holds if and only if the three-qubit state is biseparable or fully separable and that the equality of the upper bound holds if and only if the three-qubit state is the standard GHZ state. These properties imply that not only does MCP have an operational meaning in the controlled dense coding scenario but also that it can capture the genuine tripartite entanglement of three-qubit pure states. We remark that it is possible for a separable mixed three-qubit state to attain the maximal value of MCP, for example, $\rho_{123}=\frac{1}{4}(|000\rangle\langle000|+|110\rangle\langle110|+|011\rangle\langle011|+|101\rangle\langle101|)$. Furthermore, it also means that there exists a mixed state whose MCP is not zero even though the state is fully separable. The interpretation is that both classical and quantum correlations of mixed three-qubit states $\rho_{123}$ are captured by MCP. Thus, in future work, it would be worth analyzing the relation between MCP and genuine tripartite quantum and classical correlations.

Even though we have analyzed discrete variable controlled dense coding only, it is worth mentioning generalization to continuous variable (CV) systems. In fact,  CV dense coding and CV controlled dense coding have been studied \cite{braunstein2000, loock2000, jing2003, jing2002}. In order to perform CV dense coding, a two-mode squeezed vacuum is shared by Alice and Bob, and Alice encodes CV information by displacing her beam and sends it to Bob, followed by Bob's joint measurement of quadrature variables. For controlled dense coding, a GHZ-like state of CV is shared among Alice, Bob and Charlie, and Charlie performs homodyne detection on his beam. The  result of the homodyne measurement is then sent to Alice, and Alice and Bob perform CV dense coding using Charlie's measurement result. In fact, it was shown that channel capacities with and without Charlie's assistance are different \cite{jing2002}. At a first glance, it seems possible to define the CP and MCP of CV controlled dense coding in a similar manner. However, channel capacity of CV dense coding and that of discrete variable have some differences. The channel capacity of CV dense coding requires certain constraints to prevent it from diverging, which are not required for discrete variable dense coding. In addition, the channel capacities that we have defined for discrete variable controlled dense coding are the optimal ones for given states so that it can be interpreted as an intrinsic property of the state, whereas optimal channel capacities of CV states are not known in general. Thus, in order to define CP and MCP properly and generalize our results to CV systems, it is required to find a general expression of the optimal channel capacity for a given CV state. It would be an interesting future work to  generalize the CP and MCP of discrete variable controlled dense coding to CV controlled dense coding.

\section*{Acknowledgements} The authors thank Chae-Yeun Park, Hyukjoon Kwon, and Seok Hyung Lee for valuable discussions. This work was supported by a National Research Foundation of Korea (NRF) grant funded by the Korea government (MSIP) (No.\ 2010-0018295) and by the KIST Institutional Program (Project No.\ 2E26680-16-P025).

\section*{Author Contributions} C.O. and K.J. conceived the idea and proposed the study. C.O. and H.K. performed the calculations and the proofs. C.O. and H.J. analyzed and interpreted the results. All authors contributed to discussions, reviewed the manuscript, and agreed with the submission.

\section*{Competing Financial Interests} The authors declare that they have no competing financial interests.

\bibliographystyle{naturemag}

\end{document}